\documentclass[12pt]{amsart}
\usepackage[cp1251]{inputenc}
\usepackage[english]{babel}
\usepackage{amsmath}
\usepackage{amssymb}
\usepackage{amsfonts}
\usepackage{epsfig}
\usepackage{srcltx}
\usepackage{cite}
\usepackage{float}
\usepackage[ a4paper,  mag=1000, includefoot,  left=2cm, right=2cm, top=2cm, bottom=2cm, headsep=1cm, footskip=1cm ]{geometry}
\usepackage[unicode,colorlinks]{hyperref}
\hypersetup{ colorlinks=true, citecolor=blue, linkcolor=blue}

\newtheorem{Th}{Theorem}

\newtheorem{Cor}{Corollary}

\begin{document}

\thispagestyle{empty}

\title[Capture into parametric autoresonance in the presence of noise]
{Capture into parametric autoresonance in the presence of noise}

\author{Oskar Sultanov}

\address{Oskar Sultanov,
\newline\hphantom{iii} Institute of Mathematics, Ufa Scientific Center, Russian Academy of Sciences,
\newline\hphantom{iii}  112, Chernyshevsky str., Ufa, Russia, 450008}
\email{oasultanov@gmail.com}

\maketitle {\small
\begin{quote}
\noindent{\bf Abstract.}
System of differential equations describing the initial stage of the capture of oscillatory systems into the parametric autoresonance is considered. Of special interest are solutions whose amplitude increases without bound with time.
The possibility of capture the system into the autoresonance is related with the stability of such solutions.
We study the stability of autoresonant solutions with respect to persistent perturbations of white noise type, and we show that under certain conditions on the intensity of the noise, the capture into parametric autoresonance is preserved with probability tending to one.

\medskip


\end{quote}
}

\section*{Introduction}

Autoresonance is the phenomenon of continuous phase locking of nonlinear oscillator with slowly varying parametric pumping that leads to significant growth of the energy of the oscillator. This phenomenon was first suggested in the problem of acceleration of relativistic particles~\cite{V44,M45}. Later, it was observed that autoresonance occurs widely in nature and plays the important role in many problems of nonlinear physics~\cite{FGF00,FF01,LF09}. A wide range of applications requires the study of the effect of perturbations on different mathematical models of autoresonance. The influence of additive noise on the capture into the autoresonance (non-parametric) was analysed in~\cite{BFSSh09}, where the effect of perturbations was considered only at the initial time. The problem of capture into the parametric autoresonance for a quantum anharmonic oscillator with initial disturbances was discussed in~\cite{BF14}. The effect of persistent perturbations with random jumps on the stability of autoresonance models was investigated in~\cite{OS14}. In this paper we consider the deterministic model of parametric autoresonance~\cite{KM01} and we study the effect of persistent perturbations of white noise type on the captured solutions.

The paper is organized as follows. In section~\ref{sec0}, we give the mathematical formulation of the problem.
In section~\ref{sec1} we discuss the stability of autoresonant solutions with respect to perturbations of initial data. Section~\ref{sec2} deals with a more general problem of stochastic stability of a class of locally stable dynamical systems. In section~\ref{sec3} these results are applied in the study of the capture into parametric autoresonance in the presence of stochastic perturbations.

\section{Problem statement}
\label{sec0}
We consider the system of primary parametric autoresonance equations
\begin{gather}
\label{eq0}
\frac{dr}{d\tau}= r \sin \psi-\gamma r,  \ \ \frac{d\psi}{d\tau}= r-\lambda\tau + \cos \psi, \ \ \tau>0,
\end{gather}
where $0<\gamma<1$ and $\lambda>0$ are parameters.
This system appears after the averaging of equations, describing the behaviour of nonlinear oscillators in the presence of small slowly changing parametric pumping. The real-valued functions $r(\tau)$ and $\psi(\tau)$ represent the amplitude and phase shift of harmonic oscillations. Solutions with an infinitely growing amplitude $r(\tau)\approx\lambda \tau$ and bounded phase shift $\psi(\tau)=\mathcal O(1)$ as $\tau\to\infty$ correspond to the capture into the parametric autoresonance. The existence and asymptotic behaviour of captured solutions of system \eqref{eq0} were discussed in~\cite{KM01,AM05,KG07,LK08}. In this paper we study the effect of white noise on the stability of such solutions.

The asymptotic solution of system \eqref{eq0} with growing amplitude at infinity can be constructed in the following form
\begin{gather}
\label{as}
    r(\tau)={\lambda\tau}+\sum_{k=0}^{\infty} r_k \tau^{-k}, \quad \psi(\tau)=\psi_0+\sum_{k=1}^\infty \psi_k \tau^{-k},
\end{gather}
where $r_k$ and $\psi_k$ are constant coefficients. Substituting these series in system \eqref{eq0} and grouping the expressions of same power of $\tau$ give the following recurrence relations for determining the coefficients $r_k$, $\psi_k$: $\sin\psi_0=\gamma$, $\psi_1=(\cos\psi_0)^{-1}$, $r_0=-\cos\psi_0$, $r_1=-\tan\psi_0$, etc. Note that there are no free parameters in the asymptotic series. The existence of exact particular solutions of system \eqref{eq0} with constructed asymptotics follows from~\cite{AK89}. The stability of these isolated solutions determines the presence of capture into autoresonance. We show that in the case of stability, solutions with asymptotics \eqref{as} attract many other autoresonant solutions with more complicated asymptotic expansions (see Fig.~\ref{Pic1}). Note that system \eqref{eq0} has also non-resonant solutions with the slipping phase and the bounded amplitude. The existence of such solutions excludes the global stability of autoresonant solutions for all initial data. We also note that the structure of the capture region (the set of initial points such that the corresponding solutions possess the unboundedly growing amplitude) for system \eqref{eq0} remains unknown, and we do not discuss this problem here.
\begin{figure}
\centering
\includegraphics[width=0.45\textwidth]{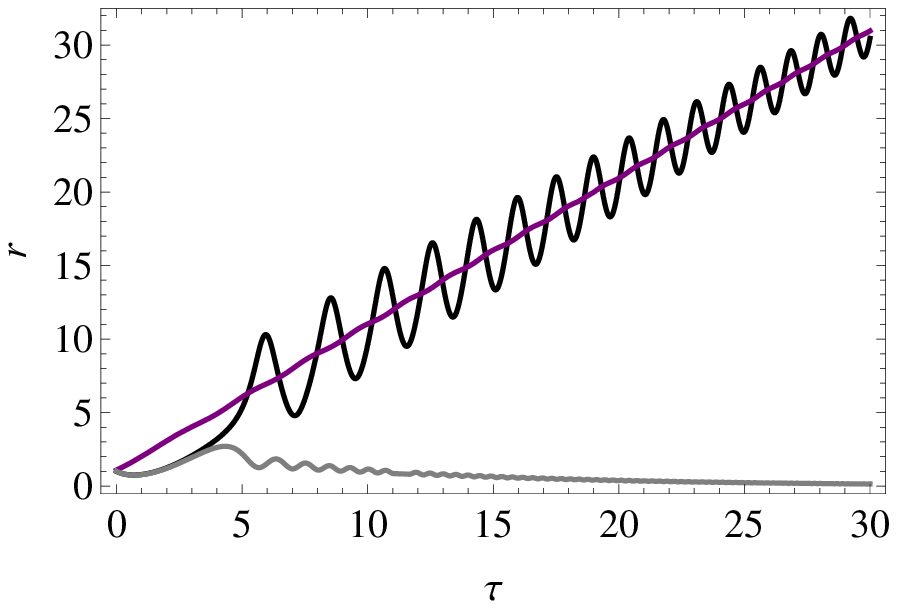} \ \ \includegraphics[width=0.45\textwidth]{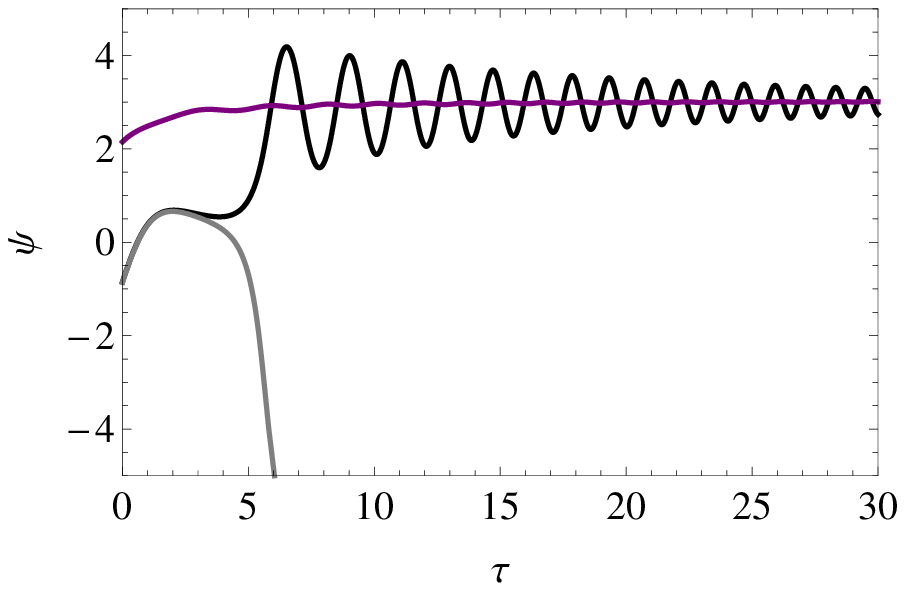}
\caption{The evolution of the amplitude $r(\tau)$ and the phase shift $\psi(\tau)$ for the solutions of \eqref{eq0} with $\lambda=1$, $\gamma=0.1$ and different initial data.}
\label{Pic1}
\end{figure}

Consider the perturbed system in the form
\begin{gather}
\label{eq1}
\frac{dr}{d\tau}=  \big[1 + \mu \xi_1(\tau)\big]r \sin \psi-\gamma r,  \ \
\frac{d\psi}{d\tau}= r-\lambda\tau + \big[1 +  \mu \xi_1(\tau)]\cos \psi+ \mu \xi_2(\tau), \quad \tau>0,
\end{gather}
where the stochastic processes $\xi_1(\tau)$ and $\xi_2(\tau)$ defined on a probability space $(\Omega, \mathcal F, \mathbb P)$  play the role of perturbations. It is assumed that $\mathbb E[\xi_i(\tau)]=0$, $\mathbb E[\xi_i(\tau)\xi_j(0)]=\delta_{ij}\sigma^2_i(\tau)\delta(\tau)$ for all $i,j\in\{1,2\}$ and $\tau\geq 0$, where $\delta_{ij}$ is the Kronecker delta, $\delta(\tau)$ is the Dirac delta function, and the deterministic functions $\sigma_i(\tau)$ together with the small parameter $0<\mu<1$ are used to control the intensity of the perturbations.
Let $\xi_i(\tau)=\sigma_i(\tau) \dot w_i(\tau)$, where $w_1(\tau)$ and $w_2(\tau)$ are independent Wiener processes. Then we can consider the perturbed system \eqref{eq1} in the form of It\^{o} stochastic differential equations. Our goal is to find constraints on the functions $\sigma_1(\tau)$, $\sigma_2(\tau)$, such that the capture into parametric autoresonance is preserved in the perturbed system with probability tending to one. Since the persistent perturbation of white noise type leads to the loss of stability of solutions for all $\tau>0$ (see~\cite{FV,RKh,PRK} and Fig.~\ref{Fig2}), we consider a weaker problem. Specifically, our goal is to find the largest possible time interval on which the stability of autoresonant solutions is preserved.
\begin{figure}
\centering
\includegraphics[width=0.5\textwidth]{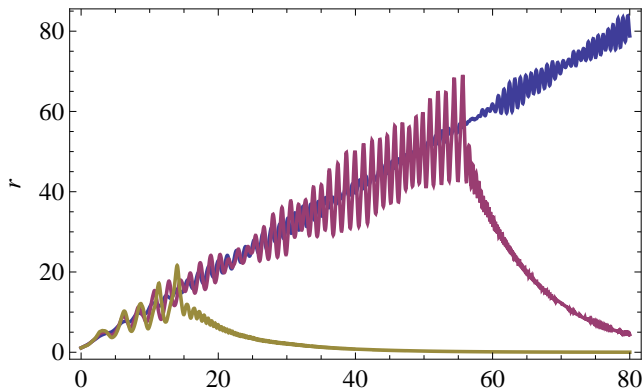}
\caption{Sample paths of the amplitude $r(\tau)$ for solutions of system \eqref{eq1} with $\lambda=1$, $\gamma=0.1$, $r(0)=1.09$, $\psi(0)=2.15$, $\sigma_1(\tau)\equiv 0$, $\sigma_2(\tau)\equiv 1$ and $\mu\in \{0.1, 0.35, 0.55\}$.}
\label{Fig2}
\end{figure}

Note that system \eqref{eq0} is of universal character in the description of parametric autoresonance in nonlinear systems. It describes long-term evolution of different nonlinear oscillations under small parametric driving. As but one example let us consider the following equation
\begin{gather}
\label{eq02}
\frac{d^2u}{dt^2}+\big(1+\varepsilon A(t) \cos 2\Phi(t) \big)\sin u +\vartheta \frac{du}{dt}= 0,
\end{gather}
where $A(t)=1+\mu \eta_1(\varepsilon t)$, $\Phi'(t)=1-\alpha t+ \mu \varepsilon \eta_2(\varepsilon t)$, $0<\varepsilon,\alpha,\vartheta,\mu\ll1$. The functions $\eta_1(s)$, $\eta_2(s)$ play the role of perturbations.
Solutions of equation \eqref{eq02} with $\mu=0$  whose amplitude increase with time from small values $|u(0)|+|u'(0)|\ll 1$ to quantities of order one are associated with the capture into parametric autoresonance. For the asymptotic description of such solutions at the initial stage of the capture we use the method of two scales.  We introduce a slow time $\tau= \varepsilon t/2$ and a fast variable  $\phi(t)=t-\alpha t^2/2$. Then the asymptotic substitution
    $$
        u(t)=  \sqrt{4 \varepsilon r(\tau)}\cos\Big(\frac{\psi(\tau)}{2}+\Phi(t)\Big)+\mathcal O(\varepsilon^{3/2})
    $$
in equation \eqref{eq02} and the averaging procedure over the fast variable $\phi(t)$ lead to system \eqref{eq1} for the slowly varying functions  $r(\tau)$ and $\psi(\tau)$ with  $\lambda=8\alpha\varepsilon^{-2}$, $\gamma=2\vartheta \varepsilon^{-1}$, $\xi_1(\tau)=\eta_1(2\tau)$, $\xi_2(\tau)=4 \eta_2(2\tau)$. In the case $\mu=0$, we get system \eqref{eq0}.

\section{Perturbations of initial data for autoresonant solutions}
\label{sec1}

Note that the unperturbed system \eqref{eq0} has two different asymptotic solutions in the form \eqref{as} distinguished by the choice of a root to the equation $\sin \psi_0=\gamma$. It can easily be checked that the solution with $\psi_0=\arcsin \gamma$ is linearly unstable. However, linear stability analysis fails for the captured solution $r_\ast(\tau)$, $\psi_\ast(\tau)$ with asymptotics \eqref{as}, $\psi_0=\pi-\arcsin \gamma$.
To study the stability of this solution, we need to take into account high-order terms of the equations.  In our analysis we use only the first terms of the asymptotic expansion for the solution,
\begin{gather}
r_\ast(\tau)=\lambda \tau+\nu +\mathcal O(\tau^{-1}), \quad \psi_\ast(\tau)=\pi-\arcsin\gamma-(\nu\tau)^{-1}+\mathcal O(\tau^{-2}), \quad \nu:=\sqrt{1-\gamma^2}.
\label{as2}
\end{gather}
We have
\begin{Th}
Suppose that the coefficients of system \eqref{eq0} satisfy the inequalities $\lambda>0$, $0<\gamma<1$. Then there exists $\tau_0>0$ and  for all $\varepsilon>0$ there exists $\delta_0>0$ such that for all $(\varrho_0, \varphi_0)$: $(\varrho_0-r_\ast(\tau_0))^2+(\varphi_0-\psi_\ast(\tau_0))^2\leq \delta_0^2$ the solution $r(\tau)$, $\psi(\tau)$ to system \eqref{eq0} with initial data $r(\tau_0)=\varrho_0$, $\psi(\tau_0)=\varphi_0$ satisfies the inequalities
\begin{gather}
\label{est}
    \sup_{\tau>\tau_0}\Big\{|r(\tau)-r_\ast(\tau)|\tau^{-1/2}\Big\}\leq\varepsilon,\quad \sup_{\tau>\tau_0}\Big\{|\psi(\tau)-\psi_\ast(\tau)|\Big\}\leq\varepsilon.
\end{gather}
\end{Th}
\begin{proof}
In system \eqref{eq0} we make the change of variables $r=r_\ast(\tau)+R(\tau)$, $\psi=\psi_\ast(\tau)+\Psi(\tau)$, and for new functions $R(\tau)$, $\Psi(\tau)$ we study the stability of the trivial solution $R(\tau)\equiv 0$, $\Psi(\tau)\equiv 0$ to the following system close to Hamiltonian system
\begin{gather}
\label{eq4}
    \frac{dR}{d\tau}=-\partial_\Psi H(R,\Psi,\tau)-\gamma R,\quad \frac{d\Psi}{d\tau}=\partial_R H(R,\Psi,\tau),
\end{gather}
where
\begin{eqnarray*}
H(R,\Psi,\tau)& = &\frac{R^2}{2}+\big(R+r_\ast(\tau)\big) \Big[\cos\big(\Psi+\psi_\ast(\tau)\big)-\cos\psi_\ast(\tau)\Big]+\Psi r_\ast(\tau) \sin\psi_\ast(\tau).
\end{eqnarray*}
By taking into account the asymptotics of the captured solution $r_\ast(\tau)$, $\psi_\ast(\tau)$ one can readily write out the asymptotics of the function $H(R,\Psi,\tau)$ as $\tau\to\infty$ and  $d=\sqrt{R^2+\Psi^2}\to 0$:
$$ H(R,\Psi,\tau)= \frac{\nu\tau\Psi^2}{2}\Big[1+\mathcal O(\Psi)\Big]+\frac{R^2}{2}+\mathcal O(d^3)+\mathcal O(d^2)\mathcal O(\tau^{-1}).$$
The asymptotic estimates are uniform with respect to $(R,\Psi,\tau)$ in the domain  $  D(d_1,\tau_1)=\{(R,\Psi,\tau)\in\mathbb R^3: d\leq d_1, \tau\geq \tau_1\}$ with positive constants $d_1>0$ and $\tau_1>1$.
It is clear that the function $H(R,\Psi,\tau)$ is positive definite function in the neighbourhood of the equilibrium $(0,0)$.
A Lyapunov function candidate for system \eqref{eq4} is constructed of the form
$$V(R,\Psi,\tau)=(\nu\tau)^{-1} \Big[ H(R,\Psi,\tau)+ \frac{\gamma R\Psi}{2} \Big].$$
From the properties of the function $H(R,\Psi,\tau)$ it follows that there exist  $0<d_0\leq d_1$ and $\tau_0\geq \tau_1$ such that
\begin{gather}
\begin{split}
\label{eq5}
\frac{1}{4}\Big[(\nu\tau)^{-1}R^2+\Psi^2\Big]\leq V(R,\Psi,\tau)\leq \frac{3}{4}\Big[(\nu\tau)^{-1}R^2+\Psi^2\Big],\\
\frac{dV}{d\tau}\Big|_{\eqref{eq4}}  = \partial_\tau V+\partial_R V [-\partial_\Psi H -\gamma R] +\partial_\psi V \partial_R H \leq \\
\leq -\frac{\gamma}{4} \Big[(\nu \tau)^{-1}R^2+\Psi^2\Big][1+\mathcal O(d)+\mathcal O(\tau^{-1})]\leq -\frac{\gamma}{6} V\leq 0
\end{split}
\end{gather}
for all $(R,\Psi,\tau)\in  D(d_0,\tau_0)$.  Integrating the last expression with respect to $\tau$, we obtain the following estimates
\begin{gather*}
  \frac{1}{4}\Big[(\nu\tau)^{-1} R^2(\tau)+\Psi^2(\tau)\Big]\leq V\big(R(\tau),\Psi(\tau),\tau)\leq V\big(R(\tau_0),\Psi(\tau_0),\tau_0)\leq \frac{3}{4\nu}\Big[R^2(\tau_0)+ \Psi^2(\tau_0)\Big]
\end{gather*}
as $\tau\geq \tau_0$, where $R(\tau)$, $\Psi(\tau)$ is the solution to system \eqref{eq4} with initial data $R^2(\tau_0)+\Psi^2(\tau_0)\leq \delta^2_0$. Therefore, for all $\varepsilon>0$ $(\varepsilon<d_0)$ there exists $\delta_0=\varepsilon\sqrt{\nu/3}>0$ such that $|R(\tau)|\tau^{-1/2}\leq \varepsilon$ and $|\Psi(\tau)|\leq \varepsilon$ for all $\tau\geq \tau_0$. By means of change of variables we derive the estimates \eqref{est}.
\end{proof}

\begin{Cor}
Suppose that the coefficients of system \eqref{eq0} satisfy the inequalities $\lambda>0$, $0<\gamma<1$. Then the solution $r_\ast(\tau)$, $\psi_\ast(\tau)$ with asymptotics \eqref{as2} is the attractor for two-parametric family of captured solutions.
\end{Cor}
\begin{proof}
From the inequality for the total derivative of the Lyapunov function $V(R,\Psi,\tau)$ along the trajectories of system \eqref{eq4} it follows that $$0\leq V(R(\tau),\Psi(\tau),\tau)\leq V(R(\tau_0),\Psi(\tau_0),\tau_0)\exp\big(-\gamma(\tau-\tau_0)/6\big)\leq \frac{3d_0^2}{4} \exp\big(-\gamma(\tau-\tau_0)/6\big)$$
as $\tau\geq \tau_0$, where $R^2(\tau_0)+\Psi^2(\tau_0)\leq d^2_0$. The change-of-variables formula implies the following asymptotic estimates for solutions to system \eqref{eq0} with initial data from the $\delta_0$-neighbourhood of the isolated autoresonant solution $r(\tau)=r_\ast(\tau)+\mathcal O(\tau^{1/2}\exp(-\gamma \tau/12))$, $\psi(\tau)=\psi_\ast(\tau)+\mathcal O(\exp(-\gamma \tau/12))$ as $\tau\geq \tau_0$.
\end{proof}

\section{Stochastic perturbations of locally stable systems}
\label{sec2}
In the study of stochastic perturbations of system \eqref{eq4} and other similar equations with locally stable solutions it is convenient to consider the system of differential equations
    \begin{equation}
    \frac{d{\bf z}}{d t}={\bf f}({\bf z}, t), \quad {\bf z}=({\bf x},{\bf y})=(x_1,\dots,x_l,y_1,\dots,y_m)\in\mathbb R^n, \quad t\geq t_0>1,
    \label{2eq1}
    \end{equation}
where $l+m=n$, $1\leq m\leq n-1$ and ${\bf f}(0, t)\equiv 0$. Suppose that the vector-valued function ${\bf f}({\bf z},t)=(f_1({\bf z},t),\dots, f_n({\bf z},t))$ is continuous and for all $T>0$ satisfies a Lipschitz condition: $|{\bf f}({\bf z}_1, t)-{\bf f}({\bf z}_2, t)|\leq M_1 |{\bf z}_1-{\bf z}_2|$ for all ${\bf z}_1, {\bf z}_2\in\mathbb R^n$, $t_0\leq  t\leq t_0+T$ with positive constant $M_1$. Assume that there exists a local Lyapunov function  $U({\bf z}, t)$ for system \eqref{2eq1} satisfying the inequalities:
    \begin{equation}
    \begin{array}{c}
        \displaystyle |{\bf x}|^2 + a t^{-b} |{\bf y}|^2\leq U({\bf z}, t)\leq A \Big[|{\bf x}|^2+a t^{-b} |{\bf y}|^2\Big], \quad |\partial_{ {\bf z}} U|^2\leq B U, \quad  |\partial_{ z_i}\partial_{z_j} U|\leq C,\\
    \displaystyle
        \frac{dU}{d t}\Big|_{\eqref{2eq1}}\stackrel{def}{=}\frac{\partial U}{\partial  t}+\sum\limits_{k=1}^{n}\frac{\partial U}{\partial z_k} f_k\leq - q U
    \end{array}
    \label{2eq2}
    \end{equation}
in the domain $\{ ({\bf z},t)\in\mathbb R^{n+1}:  |{\bf z}|\leq \rho_0,  t\geq t_0\}$ with parameters $A, B, C, q, \rho_0, b>0$, $a\geq 0$.
The existence of such Lyapunov function guaranties that the trivial solution ${\bf z}(t)\equiv 0$ is locally stable with respect to variables  ${\bf x}=(x_1,\dots,x_l)$ (if $m=0$, the trivial solution is stable with respect to all variables). Note that if $a>0$, then the solution ${\bf z}(t)\equiv 0$ is stable with respect to variables ${\bf y}=(y_1,\dots,y_m)$ in some weighted norm. Let us remark that the Lyapunov function, constructed in the previous section for the system of primary parametric autoresonance, possesses the similar estimates (cp. \eqref{eq5} with \eqref{2eq2}).
Note also that such Lyapunov functions are constructed in the stability analysis of nonlinear non-autonomous systems of differential equations (see, for example,~\cite{Vor,HK02,LK14}).

Together with system \eqref{2eq1} we consider the perturbed system in the form of It\^{o} stochastic differential equations
    \begin{gather}
    d{\bf z}(t)={\bf f}({\bf z}(t), t)\, d t+\mu\, G({\bf z}(t), t)\,d{\bf w}({ t}), \quad  {\bf z}(t_0)={\bf z}_0\in\mathbb R^n,
    \label{2eq0}
    \end{gather}
where ${\bf w}( t)=(w_1( t),\dots,w_n( t))$ is $n$-dimensional Wiener process defined on a probability space $(\Omega,\mathcal F,\mathbb P)$,
$G({\bf z}, t)=\{g_{ij}({\bf z}, t)\}_{n\times n}$ is a continuous matrix which is independent of $\omega\in\Omega$ and for all $T>0$
satisfies the following conditions $\|G({\bf z}, t)\|\leq M_2(1+|{\bf z}|)$, $\|G({\bf z}_1, t)-G({\bf z}_2, t)\|\leq M_3 |{\bf z}_1 -{\bf z}_2|$ for all
${\bf z}, {\bf z}_1, {\bf z}_2\in\mathbb R^n$, $t_0\leq t\leq t_0+T$ with positive constants $M_2, M_3>0$. We assume that ${\bf z}_0$ does not depend on $\omega\in\Omega$.
These constraints on the coefficients of system \eqref{2eq0} guarantee the existence and uniqueness of solution ${\bf z}(t)$ for all $t\geq t_0$ and for all
${\bf z}_0\in\mathbb R^n$
(see, for instance,~\cite[\S 5.2]{BO98}, \cite[\S 3.3]{RH12}).
We assume that the perturbed system does not preserve the trivial solution, $G(0,t)\not\equiv 0$. Define the class of perturbations $\mathcal A_h$ as a set of matrices $G({\bf z}, t)$ such that $|\sigma_{ij}({\bf z},t)|\leq h$ for all $|{\bf z}|\leq \rho_0$ and $t\geq t_0$, where $\sigma=G\cdot G^\ast/2=\{\sigma_{ij}\}_{n\times n}$.

We study the stability of the solution ${\bf z}(t)\equiv 0$ of system \eqref{2eq1} with respect to stochastic perturbations on a finite time interval. One variant of this approach is to find the largest possible time interval $[t_0; t_0+T_\mu]$ on which solutions to the perturbed system \eqref{2eq0} are close to the equilibrium of the deterministic system \eqref{2eq1}
(see, for instance,~\cite[Chap. 9]{FV} and~\cite[Chap. 7]{MH88}). We have
\begin{Th}\label{Th2}
Suppose that for system \eqref{2eq1} there exists a Lyapunov function $U({\bf z},t)$, possessing
estimates \eqref{2eq2}. Then, for all $N\in\mathbb N$, $h>0$, $0<\varkappa<1$, $\varepsilon_1, \varepsilon_2 > 0$ there exist $\delta, \Delta>0$ such that $\forall\, \mu<\Delta$, $G\in\mathcal A_h$, ${\bf z}_0 = ({\bf x}_0, {\bf y}_0):$ $|{\bf z}_0|<\delta$ the solution ${\bf z}(t)$ of the unperturbed system \eqref{2eq0} with initial data ${\bf z}(t_0)={\bf z}_0$ satisfies the inequalities
\begin{gather}
\label{Pest}
        \mathbb P\Big(\sup_{t_0 \leq t\leq t_0+ T_\mu}|{\bf x}( t)|\geq \varepsilon_1\Big)\leq \varepsilon_2, \quad \mathbb P\Big(\sup_{t_0 \leq t\leq t_0+ T_\mu} a t^{-b/2}|{\bf y}( t)|\geq \varepsilon_1\Big)\leq \varepsilon_2
\end{gather}
with $T_\mu=\mu^{-2N(1-\varkappa)}$.
\end{Th}
\begin{proof}
Let us fix the parameters $h>0$, $0<\varkappa<1$, $\varepsilon_2>0$ and $0<\varepsilon_1<r_0$.
Let  ${\bf z}( t)$ be a solution of system \eqref{2eq0} with $G\in\mathcal A_h$ and initial data ${\bf z}(t_0)={\bf z}_0=({\bf x}_0, {\bf y}_0)$, $|{\bf z}_0|<\delta$, and let $t_{\mathcal D}$ be the first exit time of the solution ${\bf z}( t)$ from the domain
$$\mathcal D \stackrel{def}{=}\{({\bf z}, t)\in\mathbb R^{n+1}: |{\bf z}|< \varepsilon_1, \ \  t_0 <  t < t_0+  T\}.$$
We define the function $s_t=\min\{ t_{\mathcal D}, t\}$, then ${\bf z}(s_t)$ is the process stopped at first exit time from the domain $\mathcal D$.  Positive parameters $\delta$, $T$ will be specified later.
Let us first consider the case $N=1$. The Lyapunov function for the stochastic system \eqref{2eq0} is constructed in following the form
$$
        U_1({\bf z}, t; T)=U({\bf z}, t)+ \mu^2  h n^2 C \cdot (T+t_0-t).
$$
In the study of stability of solutions to stochastic differential equations the following operator plays the role of the total derivative along the trajectories~\cite[\S 3.6]{RH12}:
$\mathcal L:=\partial_ t+\sum_{i=1}^n f_i({\bf z}, t) \partial_{z_i} + \mu^2 \sum_{i,j=1}^n \sigma_{ij}({\bf z}, t)\partial_{z_i}\partial_{z_j}.$
It easy to see that $U_1({\bf z},t;T)\geq U({\bf z},t)\geq 0$ and
   \begin{gather*}
    \mathcal L U_1  =   \frac{dU}{d t}\Big|_{\eqref{2eq1}}  +  \mu^2\sum_{i,j=1}^n \sigma_{ij} \, \partial_{z_i}\partial_{z_j} U - \mu^2 h n^2 C \leq -q  U \leq 0
    \end{gather*}
for all $({\bf z}, t)\in \mathcal D$. These estimates guarantee that $U_1({\bf z}(s_ t),s_ t)$ is a nonnegative supermartingale~\cite[\S 5.2]{RH12}. Using the properties of the function $U({\bf z},t)$ and Doob’s inequality for supermartingales, we get the following estimates
\begin{equation}\label{2eq6}
    \begin{array}{lll}
    \displaystyle \mathbb P\Big(\sup_{ t_0\leq  t\leq t_0+ T} |{\bf x}( t)|\geq \varepsilon_1\Big)
    & = & \displaystyle \mathbb P\Big(\sup_{ t_0\leq  t\leq t_0+T} |{\bf x}( t)|^2\geq \varepsilon_1^2\Big) \leq  \\
     &\leq  & \displaystyle \mathbb P\Big(\sup_{t_0\leq  t\leq t_0+T} U({\bf z}(t), t)\geq \varepsilon_1^2\Big) \leq \\
    & \leq & \displaystyle  \mathbb P\Big(\sup_{ t_0\leq  t\leq  t_0+T} U_1({\bf z}(t), t; T)\geq \varepsilon_1^2\Big) = \\
        & = & \displaystyle \mathbb P\Big(\sup_{ t\geq t_0} U_1({\bf z}(s_ t),s_t;T)\geq \varepsilon_1^2\Big) \leq \\
    & \leq & \displaystyle \frac{  U_1({\bf z}_0,t_0;T)}{\varepsilon_1^2}\leq \frac{A\big[|{\bf x}_0|^2+a t_0^{-b}|{\bf y}_0|^2\big]+\mu^2 n^2 h C T}{\varepsilon_1^2}.
    \end{array}
\end{equation}
Define $T=\mu^{-2(1-\varkappa)}$ and the parameters $\delta=({\varepsilon_1^2 \varepsilon_2}/{2 A}(1+a))^{1/2}$ and $\Delta=({ \varepsilon_1^2 \varepsilon_2 }/{2\, n^2 h \, C})^{1/2\varkappa}$. Then $A\big[|{\bf x}_0|^2+a t_0^{-b}|{\bf y}_0|^2\big] +\mu^{2\varkappa} n^2  h \, C \leq \varepsilon_1^2 \varepsilon_2$ for all $|{\bf z}_0|<\delta$, $\mu<\Delta$.
If $a=0$, this estimate holds for all ${\bf y}_0\in\mathbb R^n$. Taking into
account \eqref{2eq6}, we obtain the estimate
    \begin{gather}\label{2eq7}
    \mathbb P(\sup_{t_0\leq t\leq t_0+T}|{\bf x}( t)|\geq \varepsilon_1)\leq \varepsilon_2.
    \end{gather}

The stability for $0 \leq t\leq \mu^{-2N(1-\varkappa)}$ is proved by using the Lyapunov function $U_N({\bf z},t;T)$ in the following form~\cite{OS17}
    \begin{eqnarray*}
    U_N({\bf z},t;T) & = & \big( U({\bf z},t)\big)^N+\mu^2 a_{N-1} U_{N-1}({\bf z},t;T), \\
    U_k({\bf z},t;T) & = & \big( U({\bf z},t)\big)^k+\mu^2 a_{k-1} U_{k-1}({\bf z},t;T), \quad k=2,\dots,N-1,\\
     U_1({\bf z}, t; T)& = & U({\bf z}, t)+ \mu^2 n^2  h C \cdot (T+t_0-t),
    \end{eqnarray*}
where $a_k=(k+1) n^2 h (B+C)q^{-1}$. It is easy to check that the following inequalities hold
 \begin{eqnarray*}
    \mathcal L U_1 & \leq & -q U, \\
      \mathcal L U_2 & = & 2 U \,\mathcal L U + 2 \mu^2 \sum_{i,j=1}^n \sigma_{ij}\, \partial_{z_i}U\partial_{z_j}U + \mu^2 a_1 \mathcal L  U_1 \leq \\
                    & \leq &  - 2 q U^2 + \mu^2 \big(2 n^2 h (B+C) - a_1 q \big)   U=-2q U^2, \\
     \mathcal L U_3 & \leq & 3 U^2 \,\mathcal L U +6\mu^2 U \sum_{i,j=1}^n \sigma_{ij}\, \partial_{z_i}U\partial_{z_j}U  + \mu^2 a_2 \mathcal L  U_2 \leq \\
     & \leq & -3q  U^3+2 \mu^2 \Big(3 n^2 h\, (B+C) - a_2 q \Big) U^2 = - 3q U^3, \\
      \mathcal L U_{N}  & \leq  & N U^{N-1} \,\mathcal L U +N (N-1)\mu^2 U^{N-2} \sum_{i,j=1}^n \sigma_{ij}\, \partial_{z_i}U\partial_{z_j}U  + \mu^2 a_{N-1} \mathcal L  U_{N-1} \leq  \\
                        & \leq & -N q   U^{N}+(N-1) \mu^2  \Big(N n^2 h\, (B+C) - a_{N-1} q \Big)   U^{N-1}=- N q   U^{N}\leq 0
   \end{eqnarray*}
for all $({\bf z},t)\in\mathcal D$ and any natural number $N\geq 1$. Since $U_N({\bf z},t;T)\geq \big(U({\bf z},t)\big)^N\geq 0$ for all $({\bf z}, t)\in \mathcal D$, we see that the function $U_N({\bf z}(s_t),s_t;T)$ is a nonnegative supermartingale and the following estimates hold
     \begin{equation}\label{2eq8}
    \begin{array}{lll}
    \displaystyle \mathbb P\Big(\sup_{ t_0\leq  t\leq t_0+ T} |{\bf x}( t)|\geq \varepsilon_1\Big)
    & = & \displaystyle \mathbb P\Big(\sup_{ t_0\leq  t\leq t_0+T} |{\bf x}( t)|^{2N}\geq \varepsilon_1^{2N}\Big) \leq \\
     &\leq & \displaystyle \mathbb P\Big(\sup_{t_0\leq  t\leq t_0+T} \big(U({\bf z}(t), t)\big)^N\geq \varepsilon_1^{2N}\Big) \leq \\
    & \leq & \displaystyle  \mathbb P\Big(\sup_{ t_0\leq  t\leq  t_0+T} U_N({\bf z}(t), t; T)\geq \varepsilon_1^{2N}\Big) = \\
        & = &\displaystyle \mathbb P\Big(\sup_{ t\geq t_0} U_N({\bf z}(s_ t),s_t;T)\geq \varepsilon_1^{2N}\Big) \leq \\
    & \leq & \displaystyle \frac{  U_N({\bf z}_0,t_0;T)}{\varepsilon_1^{2N}}.
    \end{array}
    \end{equation}
Now we define $T=\mu^{-2N(1-\varkappa})$; then $|U_N({\bf z}_0,t_0;T)|\leq M_N\big[|{\bf x}_0|^{2N}+t_0^{-b N}|{\bf y}_0|^{2N}+\mu^{2N\varkappa}\big]$ as  $|{\bf z}_0|\to 0$ and $\mu\to 0$ with a positive constant $M_N$. Therefore, for all $\varepsilon_1,\varepsilon_2>0$ there exist $\delta>0, \Delta>0$ such that $U_N({\bf z}_0,t_0;T)\leq \varepsilon_1^{2N} \varepsilon_2$ for all $|{\bf z}_0|< \delta$ and $\mu<\Delta$ (if $a=0$, we can choose $|{\bf x}_0|< \delta$ and ${\bf y}_0\in\mathbb R^n$).
Taking into account \eqref{2eq8}, we obtain \eqref{2eq7}. Thus, for any natural $N$ and for all $h>0$ the trivial solution to system \eqref{2eq1} is stable with respect to variables ${\bf x}=(x_1,\dots,x_l)$ under stochastic perturbations on the interval $t_0\leq t\leq  t_0+\mu^{-2N(1-\varkappa)}$ uniformly for $G\in\mathcal A_h$. Note that if $a>0$, then from similar arguments it follows that
$$ \mathbb P\Big( \sup_{t_0\leq t\leq t_0+T} a t^{-b/2}|{\bf y}(t)|\geq \varepsilon_1\Big)\leq \varepsilon_2.$$
\end{proof}

Note that if $(1+t)^{-\beta} G\in\mathcal A_h$ with $\beta>0$ and $G\not\in\mathcal A_h$, then it can be proved that the stochastic stability of the trivial solution ${\bf z}(t)\equiv 0$ holds on the time interval $t_0\leq t\leq t_0+\mu^{(-2+\varkappa)/(1+\beta)}$. In this case, the Lyapunov function has the form $U_\beta({\bf z},t)=U({\bf z},t)+\mu^2 M_\beta (T+t_0-t)^{1+\beta}$ with a positive constant $M_\beta>0$.

\section{Stochastic perturbations of stable autoresonant solution}
\label{sec3}

In this section we study the stability of the autoresonant solution $r_\ast(\tau)$, $\psi_\ast(\tau)$ with asymptotics \eqref{as2} under stochastic perturbations. In system \eqref{eq1} we make the change of variables $r=r_\ast(\tau)+R(\tau)$, $\psi=\psi_\ast(\tau)+\Psi(\tau)$; then for the functions $R(\tau)$, $\Psi(\tau)$ we have the following system of stochastic differential equations
\begin{gather}
\begin{split}
\label{3eq1}
    & dR(\tau)=\big[-\partial_\Psi H(R,\Psi,\tau)-\gamma R\big]d\tau+\mu g_{11}(R,\Psi,\tau) d w_1(\tau), \\
    & d\Psi(\tau)= \partial_R H(R,\Psi,\tau) d\tau+\mu g_{21}(\Psi,\tau)dw_1(\tau) + \mu g_{22}(\tau)dw_2(\tau),
\end{split}
\end{gather}
where $g_{11}=\sigma_1(\tau)[r_\ast(\tau)+R]\sin(\psi_\ast+\Psi)$, $g_{21}=\sigma_1(\tau)\cos(\psi_\ast+\Psi)$, $g_{22}=\sigma_2(\tau)$.
Thus the problem is reduced to the stability analysis of the equilibrium $(0,0)$ of system \eqref{eq4} with respect to stochastic perturbations of the form \eqref{2eq0} with matrix $G=\{g_{i,j}(R,\Psi,\tau)\}$. We have
\begin{Th}
Suppose that the coefficients of system \eqref{eq0} satisfy the inequalities $\lambda>0$, $0<\gamma<1$. Then for all $N\in\mathbb N$, $h, \varepsilon_1, \varepsilon_2 > 0$ there exist $\delta, \Delta>0$ such that $\forall\, \mu<\Delta$, $(\varrho_0, \varphi_0)${\rm :} $(\varrho_0-r_\ast(\tau_0))^2+(\varphi_0-\psi_\ast(\tau_0))^2\leq \delta_0^2$,
$(\sigma_1,\sigma_2)${\rm :}   $\sup_{\tau>\tau_0}\big\{ |\sigma_1(\tau)|\tau+|\sigma_2(\tau)|\big\}\leq h$
the solution $r_\mu(\tau)$, $\psi_\mu(\tau)$ to system \eqref{eq1} with initial data $r_\mu(\tau_0)=\varrho_0$, $\psi_\mu(\tau_0)=\varphi_0$ satisfies the inequalities
\begin{gather}
\label{prest}
\begin{split}
     &   \mathbb P\Big(\sup_{\tau_0 \leq \tau\leq \tau_0+ \mu^{-N}}|\psi_\mu(\tau)-\psi_\ast(\tau)|\geq \varepsilon_1\Big)\leq \varepsilon_2, \\ & \mathbb P\Big(\sup_{\tau_0 \leq \tau\leq \tau_0+ \mu^{-N}} \tau^{-1/2}|r_\mu(\tau)-r_\ast(\tau)|\geq \varepsilon_1\Big)\leq \varepsilon_2.
        \end{split}
\end{gather}
\end{Th}
\begin{proof}
Note that system \eqref{eq4} has the Lyapunov function $V(R,\Psi,\tau)$ satisfying \eqref{2eq2} with $a=\nu^{-1}$, $b=1$ and $q=\gamma/3$.  The restrictions of the coefficients $\sigma_i(\tau)$ imply that $G=\{g_{ij}\}\in\mathcal A_h$. If we combine this with Theorem~\ref{Th2}, we get the stochastic stability of the trivial solution to system \eqref{eq4} and the estimates \eqref{Pest} with ${\bf x}=\Psi$ and ${\bf y}=R$. By means of change of variables we derive the inequalities \eqref{prest}.
\end{proof}

Thus, the stability of the isolated autoresonant solution $r_\ast(\tau)$, $\psi_\ast(\tau)$ is preserved in the perturbed system on asymptotically long time interval $\tau_0\leq \tau\leq \tau_0+\mu^{-N}$, $N\geq 1$. Therefore, the perturbation of white noise type with moderate intensity cannot destroy the stability of the capture into parametric autoresonance.

\end{document}